\documentclass[twocolumn]{revtex4-2}
\usepackage[
	bookmarks = false,
	colorlinks = true,
	linkcolor = blue,
	urlcolor= black,
	citecolor = blue]{hyperref}
\usepackage{graphicx,bm,bbm}
\usepackage{amsmath,amssymb,amsthm}
\usepackage[misc]{ifsym}
\usepackage{enumitem}

\newtheorem{lemma}{Lemma}
\newtheorem{theorem}{Theorem}

\newtheorem{proposition}{Proposition}
\newtheorem{example}{Example}
\newtheorem{definition}{Definition}

\begin{document}
\title{Annihilating and creating nonlocality without entanglement\\
by postmeasurement information}
\author{Donghoon Ha}
\author{Jeong San Kim}
\email{freddie1@khu.ac.kr}
\affiliation{Department of Applied Mathematics and Institute of Natural Sciences, Kyung Hee University, Yongin 17104, Republic of Korea}
\begin{abstract}
Nonlocality without entanglement (NLWE) is a nonlocal quantum phenomenon
that arises in separable state discrimination.
We show that the availability of the postmeasurement information about the prepared subensemble can affect the occurrence of NLWE
in discriminating nonorthogonal nonentangled states.
We provide a two-qubit state ensemble consisting of four nonorthogonal separable pure states and show that the postmeasurement information about the prepared subensemble can annihilate NLWE. We also provide another two-qubit state ensemble consisting of four nonorthogonal separable states and show that the postmeasurement information can create NLWE. Our result can provide a useful method to share or hide information using nonorthogonal separable states.
\end{abstract}
\maketitle
\section{Introduction}
\indent Whereas nonorthogonal quantum states cannot be perfectly discriminated, in general,
we can always perfectly discriminate orthogonal quantum states by using appropriate measurements \cite{chef2000,barn20091,berg2010,bae2015}.
However, it is also known that there are some multiparty orthogonal nonentangled (separable) states that cannot be perfectly discriminated only by \emph{local operations and classical communication} (LOCC) \cite{benn19991}.
In other words, there exists some
separable measurements that cannot be implemented by LOCC. In discriminating separable states of multiparty quantum systems, a phenomenon that can be achieved by global measurements but cannot be achieved only by LOCC is called \emph{nonlocality without entanglement} (NLWE) \cite{pere1991,benn19991,chit2013}.\\
\indent In discriminating orthogonal separable states, NLWE occurs when the states cannot be perfectly discriminated by LOCC \cite{benn19991,divi2003,nise2006,xu20162,hald2019,bhat2020}. 
On the other hand, in the problem of discriminating nonorthogonal separable states, NLWE occurs when the globally optimal discrimination cannot be achieved by LOCC \cite{pere1991,duan2007,chit2013,ha20211,ha20212}. For example, the double trine ensemble in a two-qubit system that consists of three nonorthogonal separable states is known to show NLWE in their discrimination \cite{pere1991,chit2013}.\\
\indent Recently, it was shown that
there are some nonorthogonal states
that can be perfectly discriminated when the \emph{postmeasurement information} (PI) about the prepared subensemble is provided \cite{akib2019}.
However, it is also known that there are some nonorthogonal states that are still impossible to be perfectly discriminated even if the PI about the prepared subensemble is available \cite{ball2008,gopa2010,carm2018}. Therefore, in discriminating multiparty nonorthogonal separable states with the PI about the prepared subensemble, NLWE occurs when the globally optimal discrimination cannot be achieved by LOCC even with the help of PI.
A natural question that can be raised here is whether the availability of the PI about the prepared subensemble can affect the occurrence of NLWE.\\
\indent Here, we provide an answer to the question by showing
that the PI about the prepared subensemble can annihilate or
create NLWE in discriminating nonorthogonal separable
states. We first consider an ensemble of two-qubit
separable states having a NLWE phenomenon and show
that the ensemble loses NLWN when the PI about the
prepared subensemble is available, thus, \emph{annihilating NLWE
by PI}. We further consider a two-qubit ensemble of separable
states without the NLWE phenomenon and show that
PI can activate NLWE of the ensemble, therefore, \emph{creating NLWE by PI}.\\
\indent This paper is organized as follows. We first recall the definition and some properties about separable states and
separable measurements in two-qubit systems. We further
recall the definition of \emph{minimum-error discrimination} (ME) \cite{hels1976,hole1979,yuen1975,bae2013},
one representative state discrimination strategy,
and provide some useful properties of ME depending on the availability of PI.
As the main results of this paper, 
we provide a two-qubit state ensemble consisting of four nonorthogonal separable states and show that NLWE occurs in discriminating the states in the ensemble. With the same ensemble, we further show that the occurrence of NLWE in the state discrimination can be vanished when the PI about the prepared subensemble is available. Moreover, we provide another two-qubit state ensemble consisting of four nonorthogonal separable states and show that NLWE does not occur in discriminating the states of the ensemble. With the same ensemble, we further show the occurrence of NLWE in the state discrimination with the PI about the prepared subensemble.
\section{Quantum state discrimination in two-qubit systems}
In two-qubit ($2\otimes2$) systems, a state is described by a density operator $\rho$, that is, a positive-semidefinite operator $\rho\succeq0$ having unit trace $\mathrm{Tr}\rho=1$, acting on a bipartite Hilbert space $\mathcal{H}=\mathbb{C}^{2}\otimes\mathbb{C}^{2}$.
A measurement with $m$ outcomes is expressed by a positive operator valued measure(POVM) $\{M_{i}\}_{i}$ that consists of $m$ positive-semidefinite operators $M_{i}\succeq0$ on $\mathcal{H}$ satisfying $\sum_{i}M_{i}=\mathbbm{1}$, where
$\mathbbm{1}$ is the identity operator on $\mathcal{H}$.
When $\{M_{i}\}_{i}$ is performed on a quantum system prepared with $\rho$, the probability that $M_{i}$ is detected is $\mathrm{Tr}(\rho M_{i})$ due to the Born rule.\\
\indent A positive-semidefinite operator is called \emph{separable} if it is a sum of positive-semidefinite product operators.
Similarly,  a measurement $\{M_{i}\}_{i}$ is called \emph{separable} if $M_{i}$ is separable for all $i$. In particular, a \emph{LOCC measurement} is a separable measurement that can be implemented by LOCC \cite{chit20142}.\\
\indent An operator $E$ on $\mathcal{H}$ is called 
\emph{positive partial transpose} (PPT) \cite{pere1996,horo1996} if 
\begin{equation}
{\rm PT}(E)\succeq0,
\end{equation}
where ${\rm PT}(\cdot)$ is the partial transposition
taken in the standard basis $\{|0\rangle,|1\rangle\}$ on the second subsystem
(Although the PPT property does not depend on the choice of subsystem to be transposed, we take the second subsystem throughout this paper for simplicity).
In two-qubit systems,
PPT is a necessary and sufficient condition for a positive-semidefinite operator to be separable \cite{horo1996}.
Thus, the set of all positive-semidefinite separable operators on $\mathcal{H}$ can be represented as
\begin{equation}\label{eq:sepset}
\mathrm{SEP}=\{E\,|\,E\succeq0,\ \mathrm{PT}(E)\succeq0\}.
\end{equation}
\indent The dual set to $\mathrm{SEP}$ is defined as
\begin{equation}\label{eq:dualsep}
\mathrm{SEP}^{*}=\{A\,|\,\mathrm{Tr}(AB)\geqslant0\ \forall B\in\mathrm{SEP}\}.
\end{equation}
Since all elements of $\mathrm{SEP}$ are positive semidefinite,
all positive semidefinite operators are in $\mathrm{SEP}^{*}$.
We also note that all PPT operators are in $\mathrm{SEP}^{*}$
because all elements of $\mathrm{SEP}$ are PPT and
$\mathrm{Tr}(AB)=\mathrm{Tr}[\mathrm{PT}(A)\mathrm{PT}(B)]$
for any two operators $A$ and $B$.\\
\indent Throughout this paper, we only consider
the situation 
of discriminating states from the state ensemble, 
\begin{equation}\label{eq:ense}
\mathcal{E}=\{\eta_{i},\rho_{i}\}_{i\in\Lambda},\ \Lambda=\{0,1,+,-\},
\end{equation}
where $\rho_{i}$ is a $2\otimes2$ separable state and $\eta_{i}$ is the probability that state $\rho_{i}$ is prepared.\\
\indent The ensemble $\mathcal{E}$ can be seen as an ensemble consisting of 
two subensembles,
\begin{equation}\label{eq:suben}
\begin{array}{ll}
\mathcal{E}_{0}=\{\eta_{i}/\sum_{j\in\mathsf{A}_{0}}\eta_{j},\rho_{i}\}_{i\in\mathsf{A}_{0}},& 
\mathsf{A}_{0}=\{\,0\,,\,1\,\},\\[1mm]
\mathcal{E}_{1}=\{\eta_{i}/\sum_{j\in\mathsf{A}_{1}}\eta_{j},\rho_{i}\}_{i\in\mathsf{A}_{1}},&
\mathsf{A}_{1}=\{+,-\},
\end{array}
\end{equation}
where $\mathcal{E}_{0}$ and $\mathcal{E}_{1}$ are prepared with probabilities $\sum_{j\in\mathsf{A}_{0}}\eta_{j}$ and $\sum_{j\in\mathsf{A}_{1}}\eta_{j}$, 
respectively.
\subsection{Minimum-error discrimination}
Let us consider 
the state discrimination of $\mathcal{E}$ in Eq.~\eqref{eq:ense} using a measurement $\{M_{i}\}_{i\in\Lambda}$
where each measurement outcome from $M_{i}$ means that the prepared state is guessed to be $\rho_{i}$.
\emph{ME of $\mathcal{E}$} is to minimize the average probability of errors that occur in guessing the prepared state.
Equivalently, ME of $\mathcal{E}$ is to maximize the average probability of 
correctly guessing the prepared state where the optimal success probability is defined as
\begin{equation}\label{eq:pgem}
p_{\rm G}(\mathcal{E})=\max_{\{M_{i}\}_{i\in\Lambda}}\sum_{i\in\Lambda}\eta_{i}\mathrm{Tr}(\rho_{i}M_{i}),
\end{equation}
over all possible POVMs.
The optimality of the POVMs in Eq.~\eqref{eq:pgem} can be confirmed by
the following necessary and sufficient condition \cite{hole1979,yuen1975,barn20092}:
\begin{equation}\label{eq:mdnsc}
\sum_{i\in\Lambda}\eta_{i}\rho_{i}M_{i}-
\eta_{j}\rho_{j}\succeq0\ 
\forall j\in\Lambda.
\end{equation}
\indent When the available measurements are limited to LOCC measurements, we denote
the maximum success probability by
\begin{equation}\label{eq:plelocc}
p_{\rm L}(\mathcal{E})=\max_{\rm LOCC}\sum_{i\in\Lambda}\eta_{i}\mathrm{Tr}(\rho_{i}M_{i}).
\end{equation}
Since the states of $\mathcal{E}$ are nonentangled, 
NLWE occurs in terms of ME if and only if ME of $\mathcal{E}$ cannot be achieved only by LOCC, that is,
\begin{equation}\label{eq:nlweme}
p_{\rm L}(\mathcal{E})<p_{\rm G}(\mathcal{E}).
\end{equation}
The following proposition provides an upper bound of $p_{\rm L}(\mathcal{E})$.
\begin{proposition}[\cite{band2015}]\label{pro:pletrh}
If $H$ is a Hermitian operator with
\begin{equation}
H-\eta_{i}\rho_{i}\in\mathrm{SEP}^{*}\ \forall i\in\Lambda,
\end{equation}
then $\mathrm{Tr}H$ is an upper bound of $p_{\rm L}(\mathcal{E})$. 
\end{proposition}
\subsection{Quantum state discrimination with postmeasurement information}
\indent In the subsection, we consider ME of $\mathcal{E}$ in Eq.~\eqref{eq:ense} when the classical information $b\in\{0,1\}$ about the prepared subensemble $\mathcal{E}_{b}$ defined in Eq.~\eqref{eq:suben}
is given after performing a measurement.
In this situation, it is known that
a measurement can be expressed by 
a POVM $\{M_{\vec{\omega}}\}_{\vec{\omega}\in\Omega}$ 
with the Cartesian product outcome space,
\begin{equation}
\Omega=\mathsf{A}_{0}\times\mathsf{A}_{1},
\end{equation}
where each $M_{(\omega_{0},\omega_{1})}$ indicates the detection of $\rho_{\omega_{0}}$ or $\rho_{\omega_{1}}$ 
according to PI $b=0$ or $1$, respectively \cite{ball2008,gopa2010}. \\
\begin{figure}[!tt]
\centerline{\includegraphics*[bb=0 130 840 460,scale=0.29]{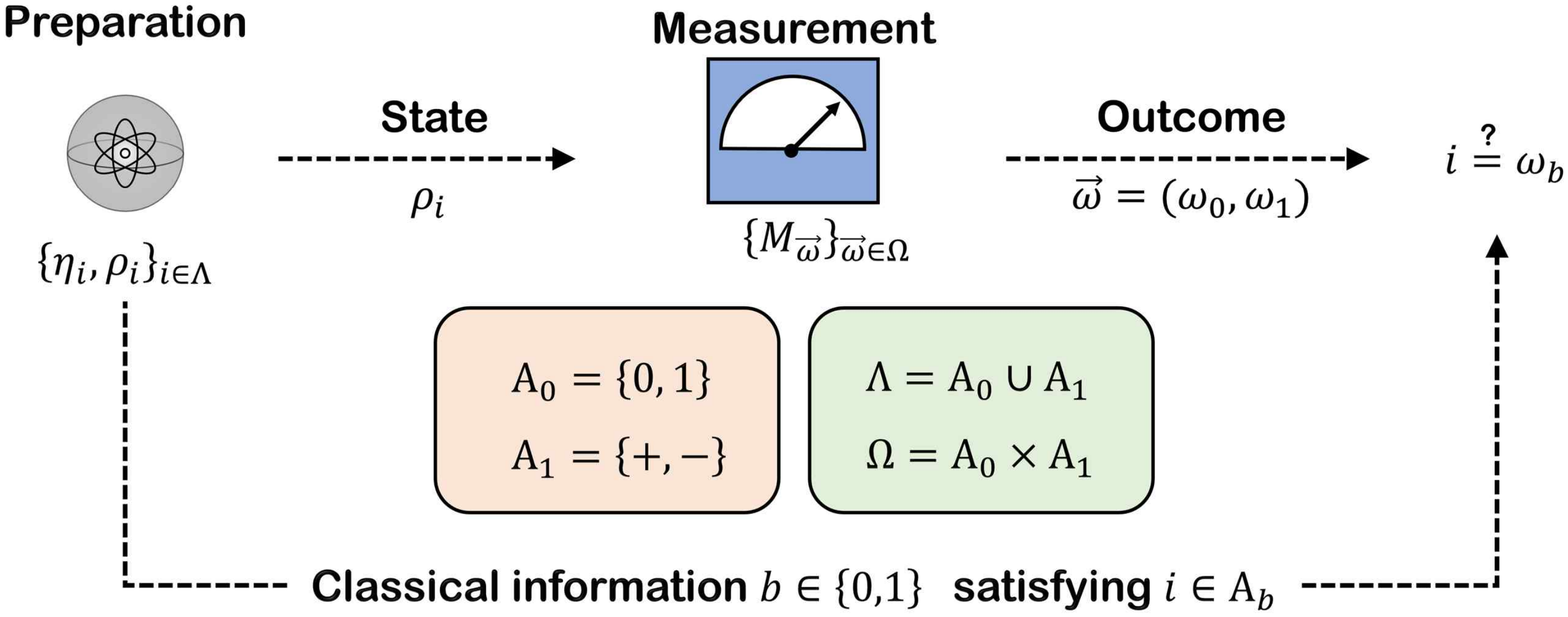}}
\caption{{\bf ME of $\mathcal{E}=\{\eta_{i},\rho_{i}\}_{i\in\Lambda}$ with PI.}
For each $i\in\Lambda$, state $\rho_{i}$ is prepared with the probability $\eta_{i}$. After performing a measurement $\{M_{\vec{\omega}}\}_{\vec{\omega}\in\Omega}$,
the classical information $b\in\{0,1\}$ satisfying $i\in\mathsf{A}_{b}$ is given.
For each measurement outcome $(\omega_{0},\omega_{1})=\vec{\omega}\in\Omega$, the prepared state is guessed to be $\rho_{\omega_{0}}$ or $\rho_{\omega_{1}}$ according to PI $b=0$ or $1$. When $\rho_{i}$ is prepared, it is correctly guessed if a measurement outcome $\vec{\omega}\in\Omega$ with $\omega_{b}=i$ is obtained; otherwise, errors occur in guessing the prepared state.
}\label{fig:post}
\end{figure}
\indent \emph{ME of $\mathcal{E}$ with PI} is
to minimize the average error probability.
Equivalently, ME of $\mathcal{E}$ with PI is to maximize the average probability of correct guessing where the optimal success probability is defined as
\begin{equation}\label{eq:pgpie}
p_{\rm G}^{\rm PI}(\mathcal{E})=
\max_{\{M_{\vec{\omega}}\}_{\vec{\omega}\in\Omega}}
\sum_{b\in\{0,1\}}\sum_{i\in\mathsf{A}_{b}}
\eta_{i}\mathrm{Tr}\Big[\rho_{i} \sum_{\substack{\vec{\omega}\in\Omega\\\omega_{b}=i}} M_{\vec{\omega}}\Big]
\end{equation}
over all possible POVMs. 
Note that when $\rho_{i}$ is prepared and PI $b\in\{0,1\}$ with $i\in\mathsf{A}_{b}$ is given, the prepared state is correctly guessed if we obtain a measurement outcome $\vec{\omega}\in\Omega$ with $\omega_{b}=i$; otherwise, errors occur in guessing the prepared state.
Figure~\ref{fig:post} illustrates ME of $\mathcal{E}$ with PI.\\
\indent We note that for a given POVM $\{M_{\vec{\omega}}\}_{\vec{\omega}\in\Omega}$, the average probability of correct guessing, that is, the right-hand side of Eq.~\eqref{eq:pgpie} without maximization, can be rewritten as 
\begin{eqnarray}
&&\sum_{b\in\{0,1\}}\sum_{i\in\mathsf{A}_{b}}
\eta_{i}\mathrm{Tr}\Big[\rho_{i} \sum_{\substack{\vec{\omega}\in\Omega\\\omega_{b}=i}} M_{\vec{\omega}}\Big]\nonumber\\
&=&\sum_{b\in\{0,1\}}\sum_{i\in\mathsf{A}_{b}}
\sum_{\substack{\vec{\omega}\in\Omega\\\omega_{b}=i}}
\mathrm{Tr}(\eta_{\omega_{b}}\rho_{\omega_{b}}M_{\vec{\omega}})\nonumber\\
&=&\sum_{b\in\{0,1\}}\sum_{\vec{\omega}\in\Omega}\mathrm{Tr}(\eta_{\omega_{b}}\rho_{\omega_{b}} M_{\vec{\omega}})\nonumber\\
&=&2\sum_{\vec{\omega}\in\Omega}
\frac{1}{2}\mathrm{Tr}\Big[\sum_{b\in\{0,1\}}\eta_{\omega_{b}}\rho_{\omega_{b}} M_{\vec{\omega}}\Big]\nonumber\\
&=&2\sum_{\vec{\omega}\in\Omega}\tilde{\eta}_{\vec{\omega}}
\mathrm{Tr}(\tilde{\rho}_{\vec{\omega}}M_{\vec{\omega}}),\label{eq:aplpie}
\end{eqnarray}
where 
\begin{equation}\label{eq:trhow}
\begin{array}{rcl}
\tilde{\eta}_{\vec{\omega}}
=\frac{1}{2}\sum_{b\in\{0,1\}}\eta_{w_{b}},\ 
\tilde{\rho}_{\vec{\omega}}
=\frac{\sum_{b\in\{0,1\}}\eta_{w_{b}}\rho_{\omega_{b}}}{\sum_{b'\in\{0,1\}}\eta_{w_{b'}}}.
\end{array}
\end{equation}\\
\indent When the available measurements are limited to LOCC measurements, we denote the maximum success probability by
\begin{eqnarray}\label{eq:plpie}
p_{\rm L}^{\rm PI}(\mathcal{E})&=&
\max_{\mathrm{LOCC}}
\sum_{b\in\{0,1\}}\sum_{i\in\mathsf{A}_{b}}
\eta_{i}\mathrm{Tr}\Big[\rho_{i} \sum_{\substack{\vec{\omega}\in\Omega\\\omega_{b}=i}} M_{\vec{\omega}}\Big].
\end{eqnarray}
Because the states in $\mathcal{E}$ are nonentangled,
NLWE occurs in terms of ME with PI if and only if ME of $\mathcal{E}$ with PI cannot be achieved only by LOCC, that is,
\begin{equation}\label{eq:defnlwemepi}
p_{\rm L}^{\rm PI}(\mathcal{E})<p_{\rm G}^{\rm PI}(\mathcal{E}).
\end{equation}
\indent Here, we note that $\{\tilde{\eta}_{\vec{\omega}}\}_{\vec{\omega}\in\Omega}$ is a set of positive numbers satisfying $\sum_{\vec{\omega}\in\Omega}\tilde{\eta}_{\vec{\omega}}=1$ and 
$\{\tilde{\rho}_{\vec{\omega}}\}_{\vec{\omega}\in\Omega}$ is a set of density operators.
Thus, Eqs.~\eqref{eq:aplpie} and \eqref{eq:plpie} imply that $p_{\rm L}^{\rm PI}(\mathcal{E})$ is twice the maximum success probability for ME of $\tilde{\mathcal{E}}$ ,
\begin{equation}\label{eq:plpie2}
p_{\rm L}^{\rm PI}(\mathcal{E})=2p_{\rm L}(\tilde{\mathcal{E}}),
\end{equation}
where $\tilde{\mathcal{E}}$ is the ensemble consisting of the average states $\tilde{\rho}_{\vec{\omega}}$ prepared with the nonzero probabilities $\tilde{\eta}_{\vec{\omega}}$ in Eq.~(\ref{eq:trhow}),
\begin{equation}\label{eq:tedef}
\begin{array}{c}
\tilde{\mathcal{E}}=\{\tilde{\eta}_{\vec{\omega}},\tilde{\rho}_{\vec{\omega}}\}_{\vec{\omega}\in\Omega}.
\end{array}
\end{equation}
\indent In the following lemma, we provide an upper bound of $p_{\rm L}^{\rm PI}(\mathcal{E})$.
\begin{lemma}\label{lem:plpietrh}
If $\tilde{H}$ is a Hermitian operator satisfying
\begin{equation}
\tilde{H}-\tilde{\eta}_{\vec{\omega}}\tilde{\rho}_{\vec{\omega}}\in\mathrm{SEP}^{*}\ \forall\vec{\omega}\in\Omega,
\end{equation}
then $2\mathrm{Tr}\tilde{H}$ is an upper bound of $p_{\rm L}^{\rm PI}(\mathcal{E})$.
\end{lemma}
\begin{proof}
For the ensemble $\tilde{\mathcal{E}}$ in Eq.~\eqref{eq:tedef}, Proposition~\ref{pro:pletrh} implies that $\mathrm{Tr}\tilde{H}$ is an upper bound of $p_{\rm L}(\tilde{\mathcal{E}})$. Thus, $2\mathrm{Tr}\tilde{H}$ is an upper bound of $p_{\rm L}^{\rm PI}(\mathcal{E})$ due to Eq.~\eqref{eq:plpie2}.
\end{proof}
\indent We close this section by providing the concept of \emph{annihilating} and \emph{creating} NLWE by PI.
\begin{definition}
For ME of an ensemble $\mathcal{E}$ in Eq.~\eqref{eq:ense}, we say that the PI $b\in\{0,1\}$ about the prepared subensemble $\mathcal{E}_{b}$ in Eq.~\eqref{eq:suben} annihilates NLWE if NLWE occurs in discriminating the states of $\mathcal{E}$ and
the availability of PI $b$ about the prepared subensemble vanishes the occurrence of NLWE, that is,
\begin{equation}
p_{\rm L}(\mathcal{E})<p_{\rm G}(\mathcal{E}),~
p_{\rm L}^{\rm PI}(\mathcal{E})=p_{\rm G}^{\rm PI}(\mathcal{E}).
\end{equation}
Also, we say that the PI $b$ about the prepared subensemble $\mathcal{E}_{b}$ creates NLWE if NLWE does not occur in discriminating the states of $\mathcal{E}$ and
the availability of PI $b$ about the prepared subensemble
releases the occurrence of NLWE, that is,
\begin{equation}
p_{\rm L}(\mathcal{E})=p_{\rm G}(\mathcal{E}),~
p_{\rm L}^{\rm PI}(\mathcal{E})<p_{\rm G}^{\rm PI}(\mathcal{E}).
\end{equation}
\end{definition}

\section{Annihilating NLWE by postmeasurement information}\label{sec:annihilating}
In this section, we consider a situation where the PI about the prepared subensemble $\mathcal{E}_{b}$ in Eq.~\eqref{eq:suben} annihilates NLWE. We first provide a specific example of a state ensemble $\mathcal{E}$ in Eq.~\eqref{eq:ense} 
and show that NLWE occurs in discriminating the states in the ensemble. With the same ensemble, we further show that the occurrence of NLWE in the state discrimination can be vanished if the PI about the prepared subensemble is available, thus, annihilating NLWE by PI.
\begin{example}\label{ex:annihilate}
Let us consider the ensemble $\mathcal{E}$ in Eq.~\eqref{eq:ense} with
\begin{equation}\label{eq:ftqs02}
\begin{array}{lcllcl}
\eta_{0}&=&\frac{\gamma}{2(1+\gamma)},&
\rho_{0}&=&|0\rangle\!\langle0|\!\otimes\!|0\rangle\!\langle0|,\\[1mm]
\eta_{1}&=&\frac{\gamma}{2(1+\gamma)},&
\rho_{1}&=&|0\rangle\!\langle0|\!\otimes\!|1\rangle\!\langle1|,\\[1mm]
\eta_{+}&=&\frac{1}{2(1+\gamma)},&
\rho_{+}&=&|+\rangle\!\langle+|\!\otimes\!|+\rangle\!\langle+|,\\[1mm]
\eta_{-}&=&\frac{1}{2(1+\gamma)},&
\rho_{-}&=&|-\rangle\!\langle-|\!\otimes\!|-\rangle\!\langle-|,
\end{array}
\end{equation}
where $2\leqslant\gamma<\infty$,  $\{|0\rangle,|1\rangle\}$ is the standard basis in one-qubit system, and
$|\pm\rangle=\frac{1}{\sqrt{2}}(|0\rangle\pm|1\rangle)$.
In this case, the subensembles in Eq.~\eqref{eq:suben} become
\begin{equation}
\begin{array}{l}
\mathcal{E}_{0}=\{\frac{1}{2},|0\rangle\!\langle0|\!\otimes\!|0\rangle\!\langle0|,\ 
\frac{1}{2},|0\rangle\!\langle0|\!\otimes\!|1\rangle\!\langle1|\},\\[1mm]
\mathcal{E}_{1}=\{\frac{1}{2},|+\rangle\!\langle+|\!\otimes\!|+\rangle\!\langle+|,\ 
\frac{1}{2},|-\rangle\!\langle-|\!\otimes\!|-\rangle\!\langle-|\},
\end{array}
\end{equation}
with the probabilities of preparation  
$\frac{\gamma}{1+\gamma}$ and $\frac{1}{1+\gamma}$, respectively.
\end{example}
\indent To show the occurrence of NLWE in terms of ME about the state ensemble $\mathcal{E}$ in Example~\ref{ex:annihilate}, we first evaluate the optimal success probability $p_{\rm G}(\mathcal{E})$ defined in Eq.~\eqref{eq:pgem}.
From the optimality condition in Eq.~\eqref{eq:mdnsc}
together with a straightforward calculation, 
we can easily verify that
the following POVM $\{M_{i}\}_{i\in\Lambda}$ is optimal for $p_{\rm G}(\mathcal{E})$:
\begin{equation}\label{eq:opm2}
\begin{array}{lcl}
M_{0}=|\Gamma_{0}\rangle\!\langle\Gamma_{0}|,\ 
M_{+}=|\mu_{+}\rangle\!\langle\mu_{+}|\otimes|+\rangle\!\langle+|,\\[1mm]
M_{1}=|\Gamma_{1}\rangle\!\langle\Gamma_{1}|,\
M_{-}=|\mu_{-}\rangle\!\langle\mu_{-}|\otimes|-\rangle\!\langle-|,
\end{array}
\end{equation}
where 
\begin{equation}
\begin{array}{rcl}
|\mu_{\pm}\rangle&=&\sqrt{\frac{1}{2}-\frac{\gamma}{2\sqrt{1+\gamma^{2}}}}|0\rangle\pm\sqrt{\frac{1}{2}+\frac{\gamma}{2\sqrt{1+\gamma^{2}}}}|1\rangle,\\[1mm]
|\Gamma_{0}\rangle&=&
\sqrt{\frac{1}{2}+\frac{\gamma}{2\sqrt{1+\gamma^{2}}}}|00\rangle-\sqrt{\frac{1}{2}-\frac{\gamma}{2\sqrt{1+\gamma^{2}}}}|11\rangle,\\[1mm]
|\Gamma_{1}\rangle&=&
\sqrt{\frac{1}{2}+\frac{\gamma}{2\sqrt{1+\gamma^{2}}}}|01\rangle-\sqrt{\frac{1}{2}-\frac{\gamma}{2\sqrt{1+\gamma^{2}}}}|10\rangle.
\end{array}
\end{equation}
Thus, the optimality of the POVM $\{M_{i}\}_{i\in\Lambda}$ in Eq.~\eqref{eq:opm2} and the definition of $p_{\rm G}(\mathcal{E})$ lead us to
\begin{equation}\label{eq:pge12}
\begin{array}{c}
p_{\rm G}(\mathcal{E})=\frac{1}{2}\Big(1+\frac{\sqrt{1+\gamma^{2}}}{1+\gamma}\Big).
\end{array}
\end{equation}
\indent In order to obtain 
the maximum success probability $p_{\rm L}(\mathcal{E})$ in Eq.~\eqref{eq:plelocc}, we consider
lower and upper bounds of $p_{\rm L}(\mathcal{E})$.
A lower bound of $p_{\rm L}(\mathcal{E})$ can be obtained from
the following POVM $\{M_{i}\}_{i\in\Lambda}$:
\begin{equation}\label{eq:loccm}
\begin{array}{ll}
M_{0}=|0\rangle\!\langle0|\otimes|0\rangle\!\langle0|,&
M_{+}=|1\rangle\!\langle1|\otimes|+\rangle\!\langle+|,
\\[1mm]
M_{1}=|0\rangle\!\langle0|\otimes|1\rangle\!\langle1|,&
M_{-}=|1\rangle\!\langle1|\otimes|-\rangle\!\langle-|,
\end{array}
\end{equation}
which gives $\frac{1}{2}(1+\frac{\gamma}{1+\gamma})$ as the success probability in discriminating the states of the ensemble $\mathcal{E}$ in Example~\ref{ex:annihilate}.
We also note that the measurement given in Eq.~\eqref{eq:loccm} can be achieved with finite-round LOCC: We perform a measurement $\{|0\rangle\!\langle0|,|1\rangle\!\langle1|\}$ on the first subsystem
and measure $\{|0\rangle\!\langle0|,|1\rangle\!\langle1|\}$ 
or $\{|+\rangle\!\langle+|,|-\rangle\!\langle-|\}$
on the second subsystem
depending on the first measurement result $|0\rangle\!\langle0|$ or $|1\rangle\!\langle1|$.
Thus, the success probability for the LOCC measurement in Eq.~\eqref{eq:loccm} is a lower bound of $p_{\rm L}(\mathcal{E})$,
\begin{equation}\label{eq:lowermd}
\begin{array}{c}
p_{\rm L}(\mathcal{E})\geqslant\frac{1}{2}\Big(1+\frac{\gamma}{1+\gamma}\Big).
\end{array}
\end{equation}
\indent To obtain an upper bound of $p_{\rm L}(\mathcal{E})$, let us consider a Hermitian operator,
\begin{eqnarray}
H=\mbox{$\frac{1}{4(1+\gamma)}(2\gamma|0\rangle\!\langle0|\otimes\sigma_{0}
+|1\rangle\!\langle1|\otimes\sigma_{0}+\sigma_{1}\otimes\sigma_{1}),$}\quad\ 
\end{eqnarray}
where $\sigma_{0}$ and $\sigma_{1}$ are the Pauli operators,
\begin{equation}\label{eq:idpauli}
\begin{array}{l}
\sigma_{0}=|0\rangle\!\langle0|+|1\rangle\!\langle1|,\\
\sigma_{1}=|0\rangle\!\langle1|+|1\rangle\!\langle0|.
\end{array}
\end{equation}
We will show that $H-\eta_{i}\rho_{i}\in\mathrm{SEP}^{*}$ for any $i\in\Lambda$, therefore $\mathrm{Tr}H$ is an upper bound of $p_{\rm L}(\mathcal{E})$ by Proposition~\ref{pro:pletrh}.\\
\indent For each $i\in\Lambda$,  $H-\eta_{i}\rho_{i}$ can be rewritten as 
\begin{equation}\label{eq:fhss}
\begin{array}{rcl}
H-\eta_{0}\rho_{0}&=&\frac{1}{4(1+\gamma)}
\big[T_{0}+|11\rangle\!\langle11|+\mathrm{PT}(T_{0})\big],
\\[2mm]
H-\eta_{1}\rho_{1}&=&\frac{1}{4(1+\gamma)}
\big[T_{1}+|10\rangle\!\langle10|+\mathrm{PT}(T_{1})\big],
\\[2mm]
H-\eta_{+}\rho_{+}&=&\frac{2\gamma-1}{4(1+\gamma)}
(\rho_{0}+\rho_{1})+\frac{1}{2(1+\gamma)}\rho_{-},
\\[2mm]
H-\eta_{-}\rho_{-}&=&\frac{2\gamma-1}{4(1+\gamma)}
(\rho_{0}+\rho_{1})
+\frac{1}{2(1+\gamma)}\rho_{+},
\end{array}
\end{equation}
where $T_{0}$ and $T_{1}$ are positive-semidefinite operators,
\begin{equation}
\begin{array}{c}
T_{0}=\gamma|01\rangle\!\langle01|+|01\rangle\!\langle10|+|10\rangle\!\langle01|+\frac{1}{2}|10\rangle\!\langle10|,\\[1mm]
T_{1}=\gamma|00\rangle\!\langle00|+|00\rangle\!\langle11|+|11\rangle\!\langle00|+\frac{1}{2}|11\rangle\!\langle11|,
\end{array}
\end{equation}
for $2\leqslant\gamma<\infty$.
In other words, each $H-\eta_{i}\rho_{i}$ in Eq.~\eqref{eq:fhss} is a sum of positive-semidefinite operators and PPT operators.
From the argument after Eq.~\eqref{eq:dualsep},
$H-\eta_{i}\rho_{i}$ is in $\mathrm{SEP}^{*}$ for each  $i\in\Lambda$, thus, 
Proposition~\ref{pro:pletrh} leads us to
\begin{equation}\label{eq:uppermd}
\begin{array}{c}
p_{\rm L}(\mathcal{E})\leqslant\mathrm{Tr}H=\frac{1}{2}\Big(1+\frac{\gamma}{1+\gamma}\Big).
\end{array}
\end{equation}
Inequalities~\eqref{eq:lowermd} and \eqref{eq:uppermd} imply
\begin{equation}\label{eq:ple12}
\begin{array}{c}
p_{\rm L}(\mathcal{E})=\frac{1}{2}\Big(1+\frac{\gamma}{1+\gamma}\Big).
\end{array}
\end{equation}
\indent From Eqs.~\eqref{eq:pge12} and \eqref{eq:ple12}, we note that there exists
a nonzero gap between
$p_{\rm G}(\mathcal{E})$ and $p_{\rm L}(\mathcal{E})$,
\begin{equation}\label{eq:plpgfg}
\begin{array}{c}
p_{\rm L}(\mathcal{E})
=\frac{1}{2}\Big(1+\frac{\gamma}{1+\gamma}\Big)<
\frac{1}{2}\Big(1+\frac{\sqrt{1+\gamma^{2}}}{1+\gamma}\Big)
=p_{\rm G}(\mathcal{E}),
\end{array}
\end{equation}
for $2\leqslant\gamma<\infty$,  thus, NLWE occurs in terms of ME in discriminating the states of the ensemble $\mathcal{E}$ in Example~\ref{ex:annihilate}.\\
\indent Now, we show that the occurrence of NLWE in Inequality~\eqref{eq:plpgfg} can be vanished when the PI about the prepared subensemble is available. 
Let us consider the following POVM $\{M_{\vec{\omega}}\}_{\vec{\omega}\in\Omega}$,
\begin{equation}\label{eq:mepim1}
\begin{array}{c}
M_{(0,+)}=|+\rangle\!\langle+|\!\otimes\!|0\rangle\!\langle0|,\
M_{(1,+)}=|+\rangle\!\langle+|\!\otimes\!|1\rangle\!\langle1|,\\[1mm]
M_{(0,-)}=|-\rangle\!\langle-|\!\otimes\!|0\rangle\!\langle0|,\
M_{(1,-)}=|-\rangle\!\langle-|\!\otimes\!|1\rangle\!\langle1|,
\end{array}
\end{equation}
which can be performed using finite-round LOCC: Two local measurements
$\{|+\rangle\!\langle+|,|-\rangle\!\langle-|\}$ and 
$\{|0\rangle\!\langle0|,|1\rangle\!\langle1|\}$
are performed on first and second subsystems, respectively.
Moreover, it is a straightforward calculation to show that
the success probability for the LOCC measurement of Eq.~\eqref{eq:mepim1} in discriminating the states in the ensemble $\mathcal{E}$ with PI is one. That is, the states in $\mathcal{E}$ can be perfectly discriminated when PI is available.\\
\indent We note that the success probability obtained from the LOCC measurement in Eq.~\eqref{eq:mepim1} is a lower bound of $p_{\rm L}^{\rm PI}(\mathcal{E})$ in Eq.~\eqref{eq:plpie}, therefore
\begin{equation}
p_{\rm L}^{\rm PI}(\mathcal{E})\geqslant1
\end{equation}
for the ensemble $\mathcal{E}$ in Example~\ref{ex:annihilate}.
Moreover, from the definitions of $p_{\rm G}^{\rm PI}(\mathcal{E})$ and $p_{\rm L}^{\rm PI}(\mathcal{E})$ in Eqs.~\eqref{eq:pgpie} and \eqref{eq:plpie}, respectively, we have
\begin{equation}
p_{\rm G}^{\rm PI}(\mathcal{E})\geqslant p_{\rm L}^{\rm PI}(\mathcal{E}).
\end{equation}
As both $p_{\rm G}^{\rm PI}(\mathcal{E})$ and $p_{\rm L}^{\rm PI}(\mathcal{E})$ are bounded above by 1, we have
\begin{equation}\label{eq:pgplping}
p_{\rm G}^{\rm PI}(\mathcal{E})= p_{\rm L}^{\rm PI}(\mathcal{E})=1.
\end{equation}
Thus, NLWE does not occur in terms ME in discriminating the states of the ensemble $\mathcal{E}$ in Example~\ref{ex:annihilate} when the PI about the prepared subensemble is available. \\
\indent Inequality~\eqref{eq:plpgfg} shows that NLWE occurs in terms of ME about the ensemble $\mathcal{E}$ in Example~\ref{ex:annihilate},
whereas Eq.~\eqref{eq:pgplping} shows that NLWE does not occur when PI is available.
Figure~\ref{fig:lme} illustrates the relative order of $p_{\rm G}(\mathcal{E})$, $p_{\rm L}(\mathcal{E})$, $p_{\rm G}^{\rm PI}(\mathcal{E})$, and $p_{\rm L}^{\rm PI}(\mathcal{E})$ for the range of $\frac{1}{3}\leqslant\eta_{0}<\frac{1}{2}$.
\begin{theorem}\label{thm:lme}
For ME of the ensemble in Example~\ref{ex:annihilate},
the PI about the prepared subensemble annihilates NLWE.
\end{theorem}
\begin{figure}[!t]
\centerline{\includegraphics*[bb=20 20 430 310,scale=0.67]{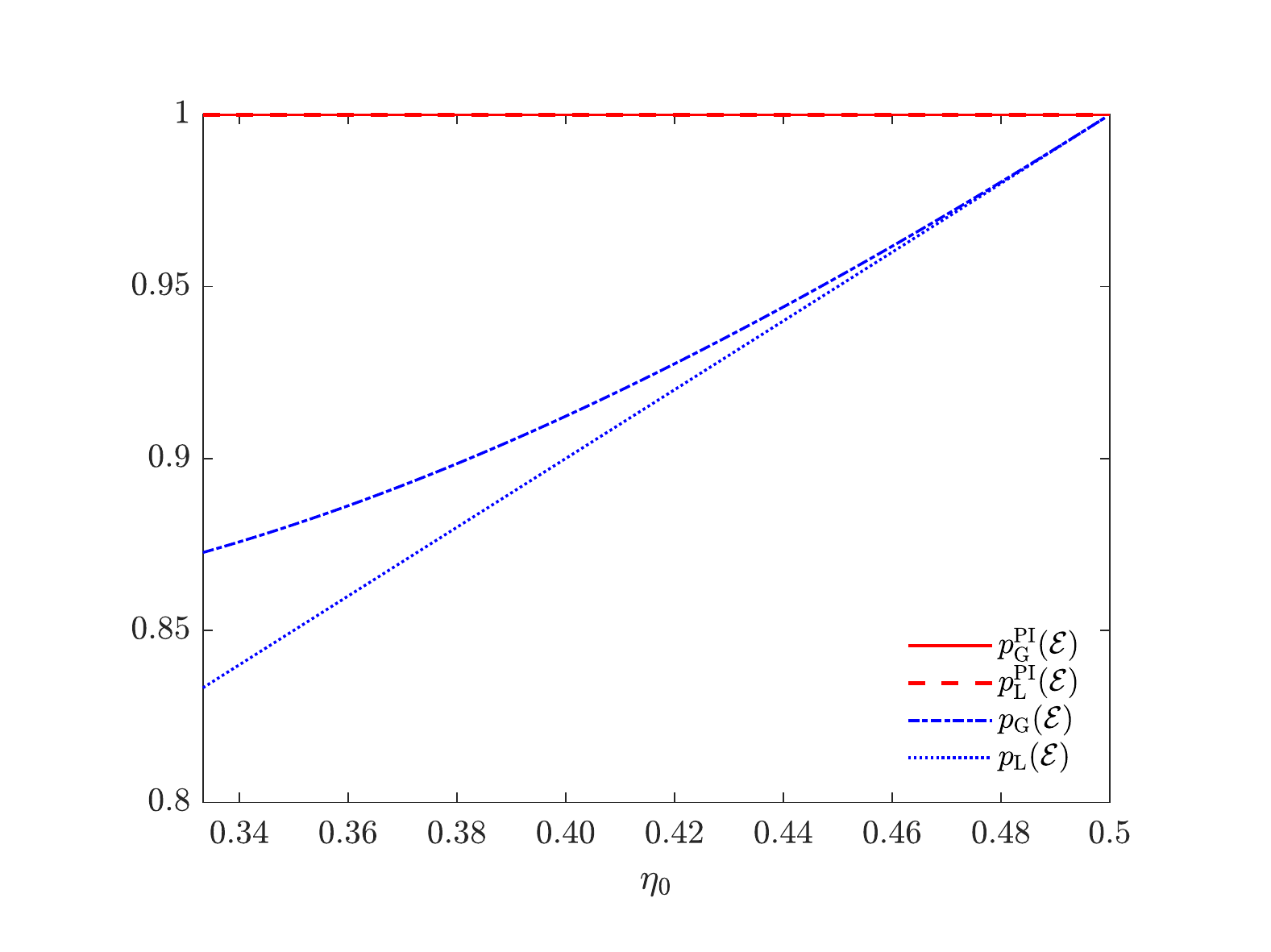}}
\caption{{\bf Annihilating NLWE by PI in terms of ME.} For all $\eta_{0}\in[\frac{1}{3},\frac{1}{2})$, 
$p_{\rm L}(\mathcal{E})$ (dotted blue line) is less than
$p_{\rm G}(\mathcal{E})$ (dot-dashed blue line), 
but $p_{\rm L}^{\rm PI}(\mathcal{E})$ (dashed red line)
is equal to $p_{\rm G}^{\rm PI}(\mathcal{E})$ (solid red line).}\label{fig:lme}
\end{figure}
\section{Creating NLWE by postmeasurement information}
In this section, we consider the opposite situation to the previous section; the PI about the prepared subensemble $\mathcal{E}_{b}$ in Eq.~\eqref{eq:suben} creates NLWE.
After providing an example of a state ensemble $\mathcal{E}$ in Eq.~\eqref{eq:ense}, we first show that NLWE does not occur in discriminating the states of the ensemble. With the same ensemble, we further show the occurrence of NLWE in the state discrimination with the help of PI, thus, creating NLWE by PI.
\begin{example}\label{ex:create}
Let us consider the ensemble $\mathcal{E}$ in Eq.~\eqref{eq:ense} with
\begin{equation}\label{eq:ftqs01}
\begin{array}{lcllcl}
\eta_{0}&=&\frac{\gamma}{2(1+\gamma)},&
\rho_{0}&=&|0\rangle\!\langle0|\!\otimes\!|0\rangle\!\langle0|,\\[1mm]
\eta_{1}&=&\frac{\gamma}{2(1+\gamma)},&
\rho_{1}&=&|0\rangle\!\langle0|\!\otimes\!|1\rangle\!\langle1|,\\[1mm]
\eta_{+}&=&\frac{1}{2(1+\gamma)},& 
\rho_{+}&=&|+\rangle\!\langle+|\!\otimes\!|+\rangle\!\langle+|,\\[1mm]
\eta_{-}&=&\frac{1}{2(1+\gamma)},& 
\rho_{-}&=&|+\rangle\!\langle+|\!\otimes\!|-\rangle\!\langle-|,
\end{array}
\end{equation}
where $2\leqslant\gamma<\infty$.
In this case, the subensembles in Eq.~\eqref{eq:suben} become
\begin{equation}
\begin{array}{l}
\mathcal{E}_{0}=\{\frac{1}{2},|0\rangle\!\langle0|\!\otimes\!|0\rangle\!\langle0|,\ 
\frac{1}{2},|0\rangle\!\langle0|\!\otimes\!|1\rangle\!\langle1|\},\\[1mm]
\mathcal{E}_{1}=\{\frac{1}{2},|+\rangle\!\langle+|\!\otimes\!|+\rangle\!\langle+|,\ 
\frac{1}{2},|+\rangle\!\langle+|\!\otimes\!|-\rangle\!\langle-|\},
\end{array}
\end{equation}
with the probabilities of preparation  
$\frac{\gamma}{1+\gamma}$ and $\frac{1}{1+\gamma}$, respectively.
\end{example}
\indent To show the nonoccurrence of NLWE in terms of ME about the ensemble $\mathcal{E}$ in Example~\ref{ex:create}, we first evaluate the optimal success probability $p_{\rm G}(\mathcal{E})$ defined in Eq.~\eqref{eq:pgem}.
From the optimality condition in Eq.~\eqref{eq:mdnsc}
together with a straightforward calculation, 
we can easily verify that
the following POVM $\{M_{i}\}_{i\in\Lambda}$ is optimal for $p_{\rm G}(\mathcal{E})$:
\begin{equation}\label{eq:fmem}
\begin{array}{lcl}
M_{0}=|\nu_{-}\rangle\!\langle\nu_{-}|\otimes|0\rangle\!\langle0|,\
M_{+}=|\nu_{+}\rangle\!\langle\nu_{+}|\otimes|+\rangle\!\langle+|,\\[1mm]
M_{1}=|\nu_{-}\rangle\!\langle\nu_{-}|\otimes|1\rangle\!\langle1|,\ 
M_{-}=|\nu_{+}\rangle\!\langle\nu_{+}|\otimes|-\rangle\!\langle-|,
\end{array}
\end{equation}
where
\begin{equation}\label{eq:nupm}
\begin{array}{c}
|\nu_{\pm}\rangle=\sqrt{\frac{1}{2}\mp\frac{\gamma}{2\sqrt{1+\gamma^{2}}}}|0\rangle\pm\sqrt{\frac{1}{2}\pm\frac{\gamma}{2\sqrt{1+\gamma^{2}}}}|1\rangle.
\end{array}
\end{equation}
Thus, the optimality of the POVM $\{M_{i}\}_{i\in\Lambda}$ in Eq.~\eqref{eq:fmem} and the definition of $p_{\rm G}(\mathcal{E})$ lead us to
\begin{equation}
\begin{array}{c}
p_{\rm G}(\mathcal{E})=\frac{1}{2}\Big(1+\frac{\sqrt{1+\gamma^{2}}}{1+\gamma}\Big).
\end{array}
\end{equation}
\indent The measurement given in Eq.~\eqref{eq:fmem} can be achieved with finite-round LOCC: First, a local measurement $\{|\nu_{+}\rangle\!\langle\nu_{+}|,|\nu_{-}\rangle\!\langle\nu_{-}|\}$
is performed on the first subsystem, and then
according to
$|\nu_{+}\rangle\!\langle\nu_{+}|$ or $|\nu_{-}\rangle\!\langle\nu_{-}|$,
a local measurement
$\{|+\rangle\!\langle+|,|-\rangle\!\langle-|\}$ or $\{|0\rangle\!\langle0|,|1\rangle\!\langle1|\}$ is performed 
on the second subsystem. Thus, the success probability for the LOCC measurement in Eq.~\eqref{eq:fmem} is a lower bound of $p_{\rm L}(\mathcal{E})$ in Eq.~\eqref{eq:plelocc}, therefore,
\begin{equation}\label{eq:plpggeq}
\begin{array}{c}
p_{\rm L}(\mathcal{E})\geqslant \frac{1}{2}\Big(1+\frac{\sqrt{1+\gamma^{2}}}{1+\gamma}\Big)
\end{array}
\end{equation}
for the ensemble $\mathcal{E}$ in Example~\ref{ex:create}. Moreover, from the definitions of $p_{\rm G}(\mathcal{E})$ and $p_{\rm L}(\mathcal{E})$ in Eqs.~\eqref{eq:pgem} and \eqref{eq:plelocc}, 
respectively, we have
\begin{equation}\label{eq:pgplgeq}
p_{\rm G}(\mathcal{E})\geqslant p_{\rm L}(\mathcal{E}).
\end{equation}
Inequalities \eqref{eq:plpggeq} and \eqref{eq:pgplgeq} lead us to
\begin{equation}\label{eq:pleqpge}
\begin{array}{c}
p_{\rm L}(\mathcal{E})=p_{\rm G}(\mathcal{E})=\frac{1}{2}\Big(1+\frac{\sqrt{1+\gamma^{2}}}{1+\gamma}\Big).
\end{array}
\end{equation}
Thus, NLWE does not occur in terms of ME in discriminating the states of the ensemble $\mathcal{E}$ in Example~\ref{ex:create}.\\
\indent Now, we show that NLWE occurs when the PI about the prepared subensemble is available. Let us consider the following POVM
$\{M_{\vec{\omega}}\}_{\vec{\omega}\in\Omega}$:
\begin{equation}\label{eq:bellm}
\begin{array}{l}
M_{(0,+)}=|\Phi_{+}\rangle\!\langle\Phi_{+}|,\,
M_{(0,-)}=|\Phi_{-}\rangle\!\langle\Phi_{-}|,\\[2mm]
M_{(1,+)}=|\Psi_{+}\rangle\!\langle\Psi_{+}|,\,
M_{(1,-)}=|\Psi_{-}\rangle\!\langle\Psi_{-}|,
\end{array}
\end{equation}
where $|\Phi_{\pm}\rangle$ and $|\Psi_{\pm}\rangle$ are Bell states,
\begin{equation}\label{eq:bells}
\begin{array}{l}
|\Phi_{\pm}\rangle=\frac{1}{\sqrt{2}}
(|00\rangle\pm|11\rangle),\\[2mm]
|\Psi_{\pm}\rangle=\frac{1}{\sqrt{2}}
(|01\rangle\pm|10\rangle).
\end{array}
\end{equation}
\indent From a straightforward calculation, we can easily see that the success probability obtained from the measurement of Eq.~\eqref{eq:bellm} in discriminating the states in the ensemble $\mathcal{E}$ with PI is one,
\begin{equation}\label{eq:pgpie1} 
p_{\rm G}^{\rm PI}(\mathcal{E})=1.
\end{equation}
That is, the states of $\mathcal{E}$ can be perfectly discriminated when PI is available.\\
\indent In order to obtain 
the maximum success probability $p_{\rm L}^{\rm PI}(\mathcal{E})$ in Eq.~\eqref{eq:plpie},
we consider lower and upper bounds of $p_{\rm L}^{\rm PI}(\mathcal{E})$. 
For a lower bound of $p_{\rm L}^{\rm PI}(\mathcal{E})$, let us first consider 
the average state ensemble $\tilde{\mathcal{E}}$ defined in Eqs.~\eqref{eq:trhow} and \eqref{eq:tedef}
with respect to Example~\ref{ex:create},
\begin{equation}\label{eq:er01pm}
\begin{array}{rcl}
\tilde{\eta}_{\vec{\omega}}&=&\frac{1}{4}\ \forall\vec{\omega}\in\Omega,\\[3mm]
\tilde{\rho}_{(0,\pm)}&=&\frac{\eta_{0}}{\eta_{0}+\eta_{\pm}}\rho_{0}+\frac{\eta_{\pm}}{\eta_{0}+\eta_{\pm}}\rho_{\pm}\\[1mm]
&=&\frac{\gamma}{1+\gamma}|0\rangle\!\langle0|\!\otimes\!|0\rangle\!\langle0|+\frac{1}{1+\gamma}|+\rangle\!\langle+|\!\otimes\!|\pm\rangle\!\langle\pm|
,\\[3mm]
\tilde{\rho}_{(1,\pm)}&=&\frac{\eta_{1}}{\eta_{1}+\eta_{\pm}}\rho_{1}+\frac{\eta_{\pm}}{\eta_{1}+\eta_{\pm}}\rho_{\pm}\\[1mm]
&=&\frac{\gamma}{1+\gamma}|0\rangle\!\langle0|\!\otimes\!|1\rangle\!\langle1|+\frac{1}{1+\gamma}|+\rangle\!\langle+|\!\otimes\!|\pm\rangle\!\langle\pm|,
\end{array}
\end{equation}
which satisfy
\begin{subequations}\label{eq:sspts}
\begin{eqnarray}
(\sigma_{0}\otimes\sigma_{2})\tilde{\rho}_{(0,\pm)}
(\sigma_{0}\otimes\sigma_{2})
&=&\tilde{\rho}_{(1,\mp)},\label{eq:sspts2}
\\
(\sigma_{0}\otimes\sigma_{1})\tilde{\rho}_{(0,\pm)}
(\sigma_{0}\otimes\sigma_{1})
&=&\tilde{\rho}_{(1,\pm)},\label{eq:sspts1}
\end{eqnarray}
\end{subequations}
with the Pauli operators $\sigma_{0}$ and $\sigma_{1}$ in Eq.~\eqref{eq:idpauli} and
\begin{equation}\label{eq:anpauli}
\begin{array}{rcl}
\sigma_{2}&=&-i|0\rangle\!\langle1|+i|1\rangle\!\langle0|.
\end{array}
\end{equation}
\indent We further consider the following Hermitian operators,
\begin{equation}
\tilde{\rho}_{(0,+)}-\tilde{\rho}_{(1,-)},\ 
\tilde{\rho}_{(1,+)}-\tilde{\rho}_{(0,-)},
\end{equation}
where both of them have the same four eigenvalues; two positive eigenvalues $\lambda_{+}$ and $\lambda_{-}$, and
two negative eigenvalues $-\lambda_{+}$ and $-\lambda_{-}$ with
\begin{equation}
\begin{array}{c}
\lambda_{\pm}=\frac{\sqrt{1+\gamma+\gamma^{2}}\pm\sqrt{1-\gamma+\gamma^{2}}}{2(1+\gamma)}
\end{array}
\end{equation}
for $2\!\leqslant\!\gamma\!<\!\infty$.
We denote $\Pi_{(0,+)}$ and $\Pi_{(1,-)}$ as
the projection operators onto the positive and negative eigenspaces of
$\tilde{\rho}_{(0,+)}-\tilde{\rho}_{(1,-)}$, respectively. Similarly, we denote
$\Pi_{(1,+)}$ and $\Pi_{(0,-)}$ as the projection operators onto the positive and negative eigenspaces of $\tilde{\rho}_{(1,+)}-\tilde{\rho}_{(0,-)}$, respectively.\\
\indent Now, we consider the following POVM $\{M_{\vec{\omega}}\}_{\vec{\omega}\in\Omega}$:
\begin{equation}\label{eq:mepim}
\begin{array}{ll}
M_{(0,+)}=\frac{1}{2}\Pi_{(0,+)},~
M_{(0,-)}=\frac{1}{2}\Pi_{(0,-)},\\[1mm]
M_{(1,+)}=\frac{1}{2}\Pi_{(1,+)},~
M_{(1,-)}=\frac{1}{2}\Pi_{(1,-)}.
\end{array}
\end{equation}
From the property of \eqref{eq:sspts2} and the definition of $\Pi_{\vec{\omega}}$, we can see that
\begin{equation}\label{eq:s0s2pi}
\begin{array}{c}
(\sigma_{0}\otimes\sigma_{2})\Pi_{(0,+)}
(\sigma_{0}\otimes\sigma_{2})
=\Pi_{(1,-)},\\[1mm]
(\sigma_{0}\otimes\sigma_{2})\Pi_{(1,+)}
(\sigma_{0}\otimes\sigma_{2})
=\Pi_{(0,-)},\\[1mm]
\Pi_{(0,+)}+\Pi_{(1,-)}=\mathbbm{1},\\[1mm]
\Pi_{(1,+)}+\Pi_{(0,-)}=\mathbbm{1}.
\end{array}
\end{equation}
Here we note that for any Hermitian operator $A$ satisfying,
\begin{equation}\label{eq:as02}
A+(\sigma_{0}\otimes\sigma_{2})A
(\sigma_{0}\otimes\sigma_{2})=\mathbbm{1},
\end{equation}
it holds that
\begin{equation}\label{eq:asr02}
\begin{array}{c}
\langle i0|A|j1\rangle
=\langle i1|A|j0\rangle
\end{array}
\end{equation}
for any $i,j\in\{0,1\}$.
From Eqs.~\eqref{eq:s0s2pi}--\eqref{eq:asr02}, we have
\begin{equation}\label{eq:ptpipm}
\mathrm{PT}(\Pi_{\vec{\omega}})=\Pi_{\vec{\omega}}\ \ \forall\vec{\omega}\in\Omega,
\end{equation}
which implies that $\Pi_{\vec{\omega}}$ is in $\mathrm{SEP}$ for any $\vec{\omega}\in\Omega$.
Thus, two POVMs $\{\Pi_{(0,+)},\Pi_{(1,-)}\}$ and $\{\Pi_{(1,+)},\Pi_{(0,-)}\}$ are separable.
Moreover, both of them can be performed using finite-round LOCC
because each of them consists of two orthogonal rank-2 projection operators  \cite{chit20141}.
The measurement given in Eq.~\eqref{eq:mepim} can be realized with finite-round LOCC 
by performing two LOCC measurements $\{\Pi_{(0,+)},\Pi_{(1,-)}\}$ and $\{\Pi_{(1,+)},\Pi_{(0,-)}\}$ 
with the equal probability $\frac{1}{2}$.\\
\indent The success probability of the LOCC measurement in Eq.~\eqref{eq:mepim} for the average state ensemble $\tilde{\mathcal{E}}$ in Eq.~\eqref{eq:er01pm} is 
\begin{equation}
\begin{array}{c}
\sum_{\vec{\omega}\in\Omega}\tilde{\eta}_{\vec{\omega}}
\mathrm{Tr}(\tilde{\rho}_{\vec{\omega}}M_{\vec{\omega}})=\frac{1}{4}\Big(1+\frac{\sqrt{1+\gamma+\gamma^{2}}}{1+\gamma}\Big).
\end{array}
\end{equation}
This probability is upper bounded by $p_{\rm L}(\tilde{\mathcal{E}})$
which is the maximum success probability for ME of $\tilde{\mathcal{E}}$
when the available measurements are limited to LOCC measurements,
\begin{equation}\label{eq:pltel}
\begin{array}{c}
p_{\rm L}(\tilde{\mathcal{E}})\geqslant\frac{1}{4}\Big(1+\frac{\sqrt{1+\gamma+\gamma^{2}}}{1+\gamma}\Big).
\end{array}
\end{equation}
Since any lower bound of $2p_{\rm L}(\tilde{\mathcal{E}})$ becomes
a lower bound of $p_{\rm L}^{\rm PI}(\mathcal{E})$ due to Eq.~\eqref{eq:plpie2}, we have
\begin{equation}\label{eq:lower1}
\begin{array}{rcl}
p_{\rm L}^{\rm PI}(\mathcal{E})\geqslant\frac{1}{2}\Big(1+\frac{\sqrt{1+\gamma+\gamma^{2}}}{1+\gamma}\Big).
\end{array}
\end{equation}
\indent To obtain an upper bound of $p_{\rm L}^{\rm PI}(\mathcal{E})$, let us first consider the following two operators,
\begin{equation}\label{eq:k0k1herm}
\begin{array}{rcl}
K_{0}&=&\frac{1}{2}\tilde{\rho}_{(0,+)}\Pi_{(0,+)}+\frac{1}{2}\tilde{\rho}_{(1,-)}\Pi_{(1,-)},\\[1mm]
K_{1}&=&\frac{1}{2}\tilde{\rho}_{(1,+)}\Pi_{(1,+)}+\frac{1}{2}\tilde{\rho}_{(0,-)}\Pi_{(0,-)}.
\end{array}
\end{equation}
Since the projective measurement $\{\Pi_{(0,+)},\Pi_{(1,-)}\}$ is optimal in ME between two states $\tilde{\rho}_{(0,+)}$ and $\tilde{\rho}_{(1,-)}$
with equal prior probability \cite{hels1976}, it satisfies a necessary and sufficient condition for a measurement 
to be optimal in ME 
between two states
$\tilde{\rho}_{(0,+)}$ and $\tilde{\rho}_{(1,-)}$ 
with equal prior probability $\frac{1}{2}$ \cite{hole1979,yuen1975,barn20092},
\begin{equation}\label{eq:kkkk01}
\begin{array}{c}
K_{0}-\frac{1}{2}\tilde{\rho}_{(0,+)}\succeq0,~
K_{0}-\frac{1}{2}\tilde{\rho}_{(1,-)}\succeq0.
\end{array}
\end{equation}
Similarly, $\{\Pi_{(1,+)},\Pi_{(0,-)}\}$ is the optimal
measurement in ME between two states $\tilde{\rho}_{(1,+)}$ and $\tilde{\rho}_{(0,-)}$
with equal prior probability $\frac{1}{2}$, thus,
\begin{equation}\label{eq:kkkk02}
\begin{array}{c}
K_{1}-\frac{1}{2}\tilde{\rho}_{(1,+)}\succeq0,~
K_{1}-\frac{1}{2}\tilde{\rho}_{(0,-)}\succeq0.
\end{array}
\end{equation}
We further note that $K_{0}$ and $K_{1}$ are Hermitian operators due to the positive semidefiniteness of \eqref{eq:kkkk01} and \eqref{eq:kkkk02}.\\
\indent Now, we consider a Hermitian operator,
\begin{equation}
\begin{array}{c}
\tilde{H}=\frac{1}{4}K_{0}+\frac{1}{4}K_{1}.
\end{array}
\end{equation}
We will show that 
$\tilde{H}-\tilde{\eta}_{\vec{\omega}}\tilde{\rho}_{\vec{\omega}}\in\mathrm{SEP}^{*}$ for all $\vec{\omega}\in\Omega$,
therefore, $2\mathrm{Tr}\tilde{H}$ is the upper bound of $p_{\rm L}^{\rm PI}(\mathcal{E})$ by Lemma~\ref{lem:plpietrh}.\\
\indent From Eqs.~\eqref{eq:sspts} and \eqref{eq:s0s2pi}, we can see that
\begin{equation}\label{eq:s0s2kk01}
\begin{array}{c}
(\sigma_{0}\otimes\sigma_{1})K_{0}(\sigma_{0}\otimes\sigma_{1})
=K_{1},\\[1mm]
(\sigma_{0}\otimes\sigma_{1})K_{1}(\sigma_{0}\otimes\sigma_{1})
=K_{0},\\[1mm]
(\sigma_{0}\otimes\sigma_{2})K_{0}(\sigma_{0}\otimes\sigma_{2})=K_{0},\\[1mm]
(\sigma_{0}\otimes\sigma_{2})K_{1}(\sigma_{0}\otimes\sigma_{2})=K_{1}.
\end{array}
\end{equation}
Moreover, for any Hermitian operator $A$ with
\begin{equation}\label{eq:as02a}
(\sigma_{0}\otimes\sigma_{2})A
(\sigma_{0}\otimes\sigma_{2})=A,
\end{equation}
it holds that
\begin{equation}\label{eq:as02ra}
\begin{array}{ll}
\langle i0|A|j0\rangle
=\langle i1|A|j1\rangle,\\[2mm]
\langle i0|A|j1\rangle
=-\langle i1|A|j0\rangle
\end{array}
\end{equation}
for any $i,j\in\{0,1\}$.
From Eqs.~\eqref{eq:s0s2kk01}--\eqref{eq:as02ra}, we have
\begin{equation}\label{eq:ptk0}
\begin{array}{c}
\mathrm{PT}(K_{0})
=(\sigma_{0}\otimes\sigma_{1})K_{0}(\sigma_{0}\otimes\sigma_{1})
=K_{1},\\[1mm]
\mathrm{PT}(K_{1})
=(\sigma_{0}\otimes\sigma_{1})K_{1}(\sigma_{0}\otimes\sigma_{1})
=K_{0}.
\end{array}
\end{equation}
Thus, for each $\vec{\omega}\in\Omega$, $\tilde{H}-\tilde{\eta}_{\vec{\omega}}\tilde{\rho}_{\vec{\omega}}$ can be rewritten as
\begin{equation}\label{eq:thtreq}
\begin{array}{rcl}
\tilde{H}-\tilde{\eta}_{(0,+)}\tilde{\rho}_{(0,+)}
&=&\frac{1}{4}(K_{0}-\frac{1}{2}\tilde{\rho}_{(0,+)})\\[1mm]
&&+\frac{1}{4}\mathrm{PT}(K_{0}-\frac{1}{2}\tilde{\rho}_{(0,+)}),
\\[2mm]
\tilde{H}-\tilde{\eta}_{(1,-)}\tilde{\rho}_{(1,-)}
&=&\frac{1}{4}(K_{0}-\frac{1}{2}\tilde{\rho}_{(1,-)})\\[1mm]
&&+\frac{1}{4}\mathrm{PT}(K_{0}-\frac{1}{2}\tilde{\rho}_{(1,-)}),
\\[2mm]
\tilde{H}-\tilde{\eta}_{(1,+)}\tilde{\rho}_{(1,+)}
&=&\frac{1}{4}\mathrm{PT}(K_{1}-\frac{1}{2}\tilde{\rho}_{(1,+)})\\[1mm]
&&+\frac{1}{4}(K_{1}-\frac{1}{2}\tilde{\rho}_{(1,+)}),
\\[2mm]
\tilde{H}-\tilde{\eta}_{(0,-)}\tilde{\rho}_{(0,-)}
&=&\frac{1}{4}\mathrm{PT}(K_{1}-\frac{1}{2}\tilde{\rho}_{(0,-)})\\[1mm]
&&+\frac{1}{4}(K_{1}-\frac{1}{2}\tilde{\rho}_{(0,-)}).
\end{array}
\end{equation}
\indent From the argument after Eq.~\eqref{eq:dualsep}
together with the positive semidefiniteness of \eqref{eq:kkkk01} and \eqref{eq:kkkk02},
each $\tilde{H}-\tilde{\eta}_{\vec{\omega}}\tilde{\rho}_{\vec{\omega}}$ in Eq.~\eqref{eq:thtreq} is in $\mathrm{SEP}^{*}$, therefore, Lemma~\ref{lem:plpietrh} leads us to
\begin{equation}\label{eq:upperb1}
\begin{array}{rcl}
p_{\rm L}^{\rm PI}(\mathcal{E})\leqslant2\mathrm{Tr}\tilde{H}
=\frac{1}{2}\Big(1+\frac{\sqrt{1+\gamma+\gamma^{2}}}{1+\gamma}\Big).
\end{array}
\end{equation}
Inequalities~\eqref{eq:lower1} and \eqref{eq:upperb1} imply
\begin{equation}\label{eq:plpi1}
\begin{array}{c}
p_{\rm L}^{\rm PI}(\mathcal{E})
=\frac{1}{2}\Big(1+\frac{\sqrt{1+\gamma+\gamma^{2}}}{1+\gamma}\Big).
\end{array}
\end{equation}
\indent From Eqs.~\eqref{eq:pgpie1} and \eqref{eq:plpi1}, we note that there exists
a nonzero gap between
$p_{\rm G}^{\rm PI}(\mathcal{E})$ and $p_{\rm L}^{\rm PI}(\mathcal{E})$,
\begin{equation}\label{eq:plleqpge}
\begin{array}{c}
p_{\rm L}^{\rm PI}(\mathcal{E})
=\frac{1}{2}\Big(1+\frac{\sqrt{1+\gamma+\gamma^{2}}}{1+\gamma}\Big)<1=p_{\rm G}^{\rm PI}(\mathcal{E}),
\end{array}
\end{equation}
for $2\leqslant\gamma<\infty$. Thus, NLWE occurs in terms of ME when the PI about the prepared subensemble is available.\\
\begin{figure}[!t]
\centerline{\includegraphics*[bb=20 20 430 310,scale=0.67]{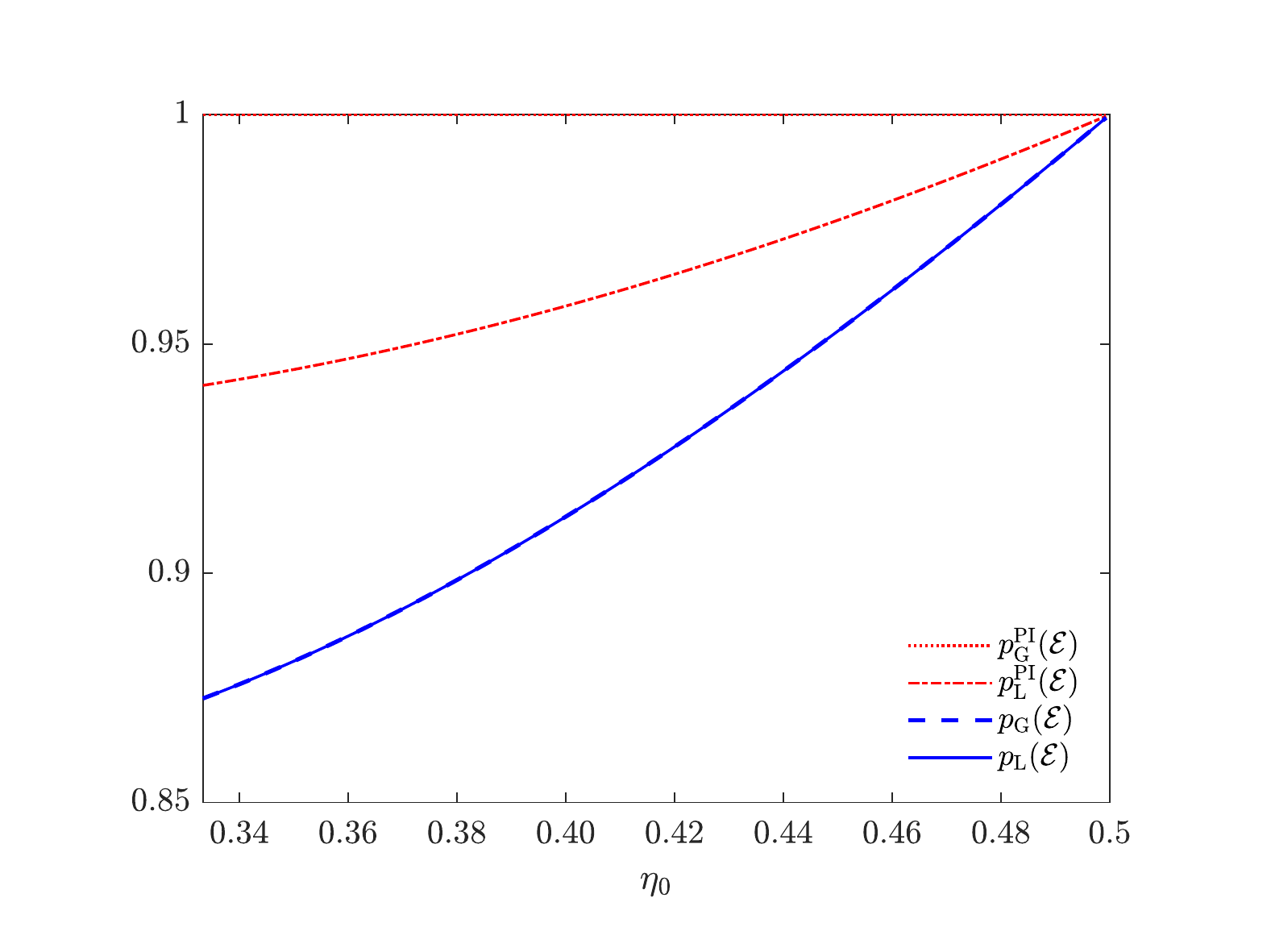}}
\caption{{\bf Creating NLWE by PI in terms of ME.} For all $\eta_{0}\in[\frac{1}{3},\frac{1}{2})$, 
$p_{\rm L}(\mathcal{E})$ (solid blue line) is equal to 
$p_{\rm G}(\mathcal{E})$ (dashed blue line), 
but $p_{\rm L}^{\rm PI}(\mathcal{E})$ (dot-dashed red line)
is less than $p_{\rm G}^{\rm PI}(\mathcal{E})$ (dotted red line).}\label{fig:ulme}
\end{figure}
\indent Equation~\eqref{eq:pleqpge} shows that NLWE does not occur in terms of ME about the ensemble $\mathcal{E}$ in Example~\ref{ex:create}, whereas Inequality~\eqref{eq:plleqpge} shows that NLWE occurs when PI is available.
Figure~\ref{fig:ulme} illustrates
the relative order of $p_{\rm G}(\mathcal{E})$, $p_{\rm L}(\mathcal{E})$, $p_{\rm G}^{\rm PI}(\mathcal{E})$, and $p_{\rm L}^{\rm PI}(\mathcal{E})$ for the range of $\frac{1}{3}\leqslant\eta_{0}<\frac{1}{2}$.
\begin{theorem}\label{thm:ulme}
For ME of the ensemble $\mathcal{E}$ in Example~\ref{ex:create}, the PI about the prepared subensemble creates NLWE.
\end{theorem}

\section{Discussion}
We have shown that the PI about the prepared subensemble can annihilate or create NLWE in discriminating multiparty nonorthogonal nonentangled quantum states. We have first provided a two-qubit state ensemble
consisting of four nonorthogonal separable states (Example~\ref{ex:annihilate}) and shown that NLWE occurs in discriminating the states in the ensemble. With the same ensemble, we have further shown that the occurrence of NLWE in the state discrimination can be vanished when the PI about the prepared subensemble is available, thus, annihilating NLWE by PI (Theorem~\ref{thm:lme}).
Moreover, we have provided another two-qubit state ensemble consisting of four nonorthogonal separable states (Example~\ref{ex:create}) and shown that NLWE does not occur in discriminating the states of the ensemble. With the same ensemble, we have further shown the occurrence of NLWE in the state discrimination with the PI about the prepared subensemble, thus, creating NLWE by PI (Theorem~\ref{thm:ulme}).\\
\indent We note that in both Examples~\ref{ex:annihilate} and \ref{ex:create}, the prepared state can be perfectly identified by a global measurement when the PI about the prepared subensemble is provided. In Example~\ref{ex:annihilate}, the prepared state can be perfectly identified by a LOCC measurement when the PI about the prepared subensemble is available. However, in Example~\ref{ex:create}, the prepared state cannot be perfectly discriminated by a LOCC measurement even if the PI about the prepared subensemble is available. As far as we know, the latter is an example exhibiting NLWE in terms of perfect discrimination with the help of PI.\\
\indent We remark that the phenomenon of creating NLWE by PI cannot arise in perfectly discriminating orthogonal separable states because there is no better state discrimination than perfect discrimination.
On the other hand, the phenomenon of annihilating NLWE by PI can arise in perfectly discriminating orthogonal separable states with local indistinguishability, such as an \emph{unextendible product basis} (UPB) \cite{benn19992}.\\
\indent For example, let us consider a two-qutrit state ensemble $\{\frac{1}{5},\rho_{i}\}_{i=1}^{5}$ consisting of UPB states $\rho_{i}$ with the equal prior probability $\frac{1}{5}$ \cite{benn19992},
\begin{equation}
\begin{array}{l}
\rho_{1}=|\phi_{1}\rangle\!\langle\phi_{1}|\otimes|2\rangle\!\langle2|,~
\rho_{4}=|0\rangle\!\langle0|\otimes|\phi_{1}\rangle\!\langle\phi_{1}|,\\[2mm]
\rho_{2}=|\phi_{2}\rangle\!\langle\phi_{2}|\otimes|0\rangle\!\langle0|,~
\rho_{5}=|2\rangle\!\langle2|\otimes|\phi_{2}\rangle\!\langle\phi_{2}|,\\[2mm]
\rho_{3}=|\phi_{3}\rangle\!\langle\phi_{3}|\otimes|\phi_{3}\rangle\!\langle\phi_{3}|,
\end{array}
\end{equation}
where $\{|0\rangle,|1\rangle,|2\rangle\}$ is the standard basis in one-qutrit system and
\begin{equation}
\begin{array}{l}
|\phi_{1}\rangle=\frac{1}{\sqrt{2}}(|0\rangle-|1\rangle),\\[2mm]
|\phi_{2}\rangle=\frac{1}{\sqrt{2}}(|1\rangle-|2\rangle),\\[2mm]
|\phi_{3}\rangle=\frac{1}{\sqrt{3}}(|0\rangle+|1\rangle+|2\rangle).
\end{array}
\end{equation}
Since every UPB can be perfectly discriminated by global measurements but cannot be perfectly discriminated only by LOCC \cite{divi2003,deri2004}, NLWE occurs in terms of the perfect discrimination of $\{\frac{1}{5},\rho_{i}\}_{i=1}^{5}$. 
However, the occurrence of NLWE can be vanished by PI because the prepared state can be perfectly identified in the following situation: The classical information on whether the prepared state belongs to $\{\rho_{1},\rho_{2},\rho_{3}\}$ or $\{\rho_{4},\rho_{5}\}$ is provided after a LOCC measurement $\{M_{(i,j)}\}_{i,j=1}^{3}$,
\begin{equation}
M_{(i,j)}=|\phi_{i}\rangle\!\langle\phi_{i}|\otimes|\phi_{j}\rangle\!\langle\phi_{j}|,~i,j=1,2,3,
\end{equation}
where each $M_{(i,j)}$ indicates the detection of $\rho_{i}$ or $\rho_{3+j}$ depending on whether the set to which the prepared state belongs is $\{\rho_{1},\rho_{2},\rho_{3}\}$ or $\{\rho_{4},\rho_{5}\}$. 
Thus, annihilating NLWE by PI.\\
\indent Our result can provide a useful method to share or hide information using nonorthogonal separable states~\cite{terh20011,divi2002,egge2002,raha2015,wang20171,bani2021}. In Example~\ref{ex:annihilate}, the PI about the prepared subensemble makes the information locally accessible, and the information can be locally shared between parties. 
On the other hand, in Example~\ref{ex:create}, the PI about the prepared subensemble makes the information
globally accessible but not locally, and the globally accessible information can be locally hidden to some extent.
Our results can also be applied to multiparty secret sharing, such as two-qubit nonlocal bases with multicopy adaptive local distinguishability \cite{bani2021}.
We finally remark that it would be an interesting future task to investigate if
the availability of PI affects the occurrence of NLWE in terms of other optimal discrimination strategies besides ME.
\section*{Acknowledgements}
This work was supported by Basic Science Research Program (Program No. NRF-2020R1F1A1A010501270) and Quantum Computing Technology Development Program (Program No. NRF-2020M3E4A1080088) through the National Research Foundation of Korea(NRF) Grant funded by the Korea government (Ministry of Science and ICT).

\end{document}